\documentclass[pra,aps,epsf,nofootinbib,
notitlepage,showpacs,10pt]{revtex4-1}

\usepackage{color}
\usepackage{float}
\usepackage{graphicx}
\usepackage{subcaption}
\usepackage{amsmath}
\usepackage{amsthm}
\usepackage{mathtools}
\newtheorem{theorem}{Theorem}
\theoremstyle{definition}

\usepackage{bm}

\newcommand{\ket}[1]{{\left| #1 \right\rangle}}

\theoremstyle{definition}
\usepackage{graphicx} 
\graphicspath{ {rrImages/} } 
\usepackage{tikz} 
\usetikzlibrary{matrix,positioning}
\usepackage{amsfonts}
\usepackage{amssymb}

\begin{document}

\markboth{James Ostrowski}
{Lower Bounds on Circuit Depth of the Quantum Approximate Optimization Algorithm}

\title{Lower Bounds on Circuit Depth of the Quantum Approximate Optimization Algorithm}

\author{James Ostrowski}
\email{jostrows@utk.edu}
\affiliation{
	Department of Industrial and Systems Engineering, University of Tennessee at Knoxville\\Knoxville, Tennessee  37996-2315 USA}
	
\author{Rebekah Herrman}
\affiliation{
	Department of Industrial and Systems Engineering, University of Tennessee at Knoxville\\Knoxville, Tennessee  37996-2315 USA}

\author{Travis S. Humble}
\affiliation{
	Quantum Computing Institute\\ Oak Ridge National Laboratory\\ Oak Ridge, Tennessee 37830 USA}
	
	
\author{George Siopsis}
\affiliation{
	Department of Physics and Astronomy, University of Tennessee at Knoxville\\Knoxville, Tennessee 37996-1200 USA}
	

\begin{abstract}
The quantum approximate optimization algorithm (QAOA) is a method of approximately solving combinatorial optimization problems. While QAOA is developed to solve a broad class of combinatorial optimization problems, it is not clear which classes of problems are best suited for it. One factor in demonstrating quantum advantage is the relationship between a problem instance and the circuit depth required to implement the QAOA method. As errors in NISQ devices increases exponentially with circuit depth, identifying lower bounds on circuit depth can provide insights into when quantum advantage could be feasible. Here, we identify how the structure of problem instances can be used to identify lower bounds for circuit depth for each iteration of QAOA and examine the relationship between problem structure and the circuit depth for a variety of combinatorial optimization problems including MaxCut and MaxIndSet. Specifically, we show how to derive a graph, $G$, that describes a general combinatorial optimization problem and show that the depth of circuit is at least the chromatic index of $G$. By looking at the scaling of circuit depth, we argue that MaxCut, MaxIndSet, and some instances of Vertex Covering and Boolean satisifiability problems are suitable for QAOA approaches while Knapsack and Traveling Sales Person problems are not. 

\end{abstract}

\maketitle


\section{Introduction}

In 2014, Farhi, Goldstone, and Gutmann introduced the quantum approximate optimization algorithm (QAOA) to approximately solve combinatorial optimization problems \cite{farhi2014quantum}. In classical combinatorial optimization, problems are defined by $n$ bits and $m$ clauses. To solve optimization problems using QAOA, the clauses are converted to Hamiltonians, and the state of the graph is initially $\ket{s}=\frac{1}{\sqrt{2^n}} \Sigma_{z} \ket{z}$, where $\{\ket{z}\}$ is the computational basis. For $p \in \mathbb{N}$, the $p$-level QAOA requires $2p$ angles, $\bm{\vec{\gamma}}=(\gamma_1, ..., \gamma_p)$ and $\bm{\vec{\beta}}=(\beta_1, ... , \beta_p)$ and alternates between the mixing Hamiltonian, $B$, and the problem Hamiltonian, $C$, to generate the state 
\begin{equation*}
    \ket{\psi(\bm{\vec{\gamma}}, \bm{\vec{\beta})}} = U(B, \beta_p)U(C, \gamma_p)...U(B, \beta_1)U(C, \gamma_1)\ket{s}
\end{equation*}
 where $U(A,\phi)= e^{-iA\phi}$. $B$ and $C$ depend on the problem of interest and the angles that maximize them can be found using classical preprocessing \cite{guerreschi2017practical, streif2019training, shaydulin2019multistart}.
 
Previously, QAOA has been used to solve bounded constraint problems \cite{farhi2014bounded} and has been studied on near-term devices \cite{zhou2018quantum}. Additionally, it has been used to look at lattice protein folding \cite{fingerhuth2018quantum}, the Max-k vertex cover problem \cite{cook2019quantum} and inspired an approach for solving linear systems using quantum computing \cite{huang2019near}. MaxCut and maximum independent set are examples of two problems that have been well studied with QAOA \cite{saleem2020, wang2018quantum, crooks2018performance, guerreschi2019qaoa}. Both can be represented as quadratic unconstrained problems, 
otherwise known as a QUBO. 
It has also been shown to exhibit a form of computational advantage in the sense that the output of low depth circuits cannot be efficiently classically simulated \cite{farhisupremacy} and general strategies have been studied for implementing it on hardware graphs \cite{wang2019xy}. 

In this paper we investigate the potential of using quantum computing to solve combinatorial optimization problems of the form
\begin{align}\label{eq:co}
& \min \  c(x) \\
& \mbox{s.t. }  p_i (x) \leq b_i & \forall i \in P\\
& x \in \{0,1\}^n
\end{align}
 \noindent where both $p_i$, contained in the collection of polynomial constraints $P$, and $c$ are polynomial functions in $\mathbb{R}^n[x_1, x_2, ... , x_n ]$ and $b_i \in \mathbb{R}$.

We identify the relationship between combinatorial optimization problems and the corresponding depth of circuit for QAOA approaches to solving this problem. Xue, Chen, Wu, and Guo showed that the cost function for QAOA decreases with the number of gates and level of noise in NISQ devices \cite{xue2019effects}, so in this paper, we specifically focus on circuit depth, although an equally important component of the fidelity of a solution is the number of iterations needed.  We only look at a single iteration because we consider all iterations have the same depth.
 
 In Sec. \ref{background}, we define graph theory terms that will be used throughout the paper. Next, in Sec. \ref{pubomapping}, we discuss how to map arbitrary combinatorial optimization problems to polynomial unconstrained binary optimization problems (PUBOs) by dualizing constraints, and apply the method to MaxCut, Maximum Independent Set, and a general combinatorial optimization problem. Additionally, we discuss how use the PUBOs to derive a hypergraph that represents a specific optimization problem and show that one plus the chromatic index of the hypergraph is equal to the depth of QAOA circuit needed to run a combinatorial optimization problem. Using this result, we analyze the depth of circuit for the MaxCut, Maximum Independent Set, and general combinatorial optimization problems. We then consider Vertex Covering, Knapsack, Traveling SalesPerson, and Boolean satisfiability problems, determine the depth of circuit required to use QAOA to solve them, and discuss the feasibility of performing them on NISQ devices in Sec. \ref{problemanalysis}. Finally, in Sec. \ref{conclusion}, we discuss avenues for future work.

\section{Background}\label{background}

 In this section, we define graph theory terms that will be used in upcoming sections. An \emph{edge coloring} of a simple graph $G=(V,E)$ is a labeling $f: E \longrightarrow [k]$, where each number represents a color. An edge coloring is \emph{proper} if for all edges $uv$ and $xv$, $f(uv) \neq f(xv)$. The smallest number of colors needed for a proper coloring of $G$ is the edge chromatic number, sometimes referred to as the \emph{chromatic index}, denoted $\chi '(G)$, and we say all edges with the same label belong to the same color class. A well known result by Vizing states that $\chi ' (G) \in \{\Delta, \Delta+1\}$, where $\Delta$ is the maximum degree of $G$ \cite{vizing1964estimate}. 

A \emph{hypergraph} $H= (V_H,E_H)$ is a generalization of a graph in which an edge may join more than two vertices. If there are $n$ vertices in $H$, then $E \subset P \setminus \{\emptyset\}$. $H$ is \emph{linear} if two edges share at most one vertex, and it is \emph{k-uniform} if all edges contain exactly $k$ vertices. A \emph{hypergraph clique} is a collection of edges, $H_c \subset E_H$, such that every element of $H_c$ is pairwise intersecting. A \emph{proper hypergraph edge coloring} is analogous to an edge coloring of a graph in that if a vertex is contained in multiple edges, they all receive distinct colors.


\section{Mapping arbitrary combinatorial optimization problems to PUBO}\label{pubomapping}

When considering combinatorial optimization problems, we will use the method of dualizing constraints to solve them and analyze circuit depth. Other methods may give different results. Consider a constraint $p_i(x) \leq b_i$, where $x=\{x_1, ..., x_n\}$. We can {\em dualize} this constraint by penalizing any solution $x'$ with $p_i(x') \geq b_i$ as follows. Let $\underline{p}_i = \min\limits_{x \in \{0,1\}^n} p_i(x)$. The ``most feasible'' solution with respect to constraint $i$ is going to be $feas_i = \underline{p}_i - b_i$ away from the constraint.  Let $k_i = \lceil \ln{feas_i} \rceil$. We can omit constraint $i$ from the set of constraints and add the term

\begin{equation}\label{omiti}
    \lambda_i\left(p_i(x) + \sum_{j  \in [k_i]} 2^j \delta_{ij} - b_i\right)^2
\end{equation}

\noindent where $\lambda_i$ is any large, positive parameter penalizing violation of constraint $i$ and $\delta_{ij}$ are additional binary variables. 

In multiplying out the above constraint, we get
\begin{equation}\label{expanded}
 \lambda_i\left(p_i(x)^2 + 2\sum_{j = \in [k_i]} 2^j p_i(x) \delta_{ij} - 2b_i p_i(x) +\left(\sum_{j\in [k_i]} 2^j \delta_{ij}\right)^2 - 2\sum_{j \in [k_i]} 2^j b_i \delta_{ij} + b_i^2\right).
 \end{equation}

The cost of this transformation is an increase in the potentially large number of new $\delta_{ij}$ variables.

Using the above process, we can write any combinatorial optimization problem of type~\eqref{eq:co} as

\begin{equation}\label{generalcoproblem}
\min_{x \in \{0,1\}^n, \delta \in \{0,1\}^{k_i} }  c(x) + \sum_{i \in P} \lambda_i\left(p_i(x)^2 + 2\sum_{j  \in [k_i]} 2^j p_i(x) \delta_{ij} - 2b_i p_i(x) +\left(\sum_{j\in [k_i]} 2^j \delta_{ij}\right)^2 - 2\sum_{j \in [k_i]} 2^j b_i \delta_{ij} + b_i^2\right).
\end{equation}

If all $p_i$ constraints are linear and $c$ is quadratic, the resulting unconstrained problem is a QUBO. Simplifying notation, we can think of a combinatorial optimization problem as the sum of monomials of the polynomial $p_i(x)$, 

$$ \min_{x\in \{0,1\}, \delta \in \{0,1\}^{k_i} }  \sum_{m_i \in \ M_i} m_i(x,\ \delta), $$
\noindent where $M_i$ is the set of monomials of $p_i(x)$ 

\subsubsection{Examples}\label{examplespartone}
In this section, we give examples of problems and how to map them to PUBOs. 

\textbf{Example: MaxCut}\label{pubomaxcutsetup}

In the combinatorial optimization problem MaxCut, the vertices of a graph, $G=(V,E)$, are partitioned into two sets such that the number of edges with an end point in each set is maximized. This problem can be formulated as
 $$\min_{x \in \{0,1\}^n} \sum_{ij \in E(G)} x_j(x_i-1) + x_i(x_j-1)= \min_{x \in \{0,1\}^n} \sum_{ij \in E(G)} 2x_ix_j - x_i - x_j$$

Note \noindent that $P= \{ \emptyset \}$, so there are no $\delta_{ij}$ terms when the problem is dualized. For example, consider the wheel graph on six vertices, $W_6$, as seen in Figure \ref{wheel3}. 

In this example, we want to minimize 

\begin{align*}
&2x_1x_2 - x_1 - x_2 + 2x_1x_3 - x_1 - x_3 + 2x_1x_4 - x_1 - x_4 + \\
&2x_1x_5 - x_1 - x_5 + 2x_1x_6 - x_1 - x_6 + 2x_2x_3 - x_2 - x_3 + \\
&2x_2x_6 - x_2 - x_6 + 2x_3x_4 - x_3 - x_4 + 2x_4x_5 - x_4 - x_5 + 2x_5x_6 - x_5 - x_6
\end{align*}
\noindent where $x_i \in \{0,1\}$ for all $i \in [6]$.

\textbf{Example: Maximum Independent Set}\label{pubomaxindsetsetup}

\begin{figure}
 \centering
 \begin{tikzpicture}[scale=2]
\begin{scope}[every node/.style={scale=1,circle,draw}]
    \node (A) at (0,0) {$v_1$};
    \node (B) at (1.25,.5) {$v_2$};
	\node (C) at (0,1.25) {$v_3$}; 
	\node (D) at (-1.25,.5) {$v_4$};
	\node (E) at (-.75,-1) {$v_5$};
	\node (F) at (.75,-1) {$v_6$};
	
\end{scope}

\draw  [blue, dashed]  (A) -- (B);
\draw  [green]  (A) -- (C);
\draw  [densely dotted, red]  (A) -- (D);
\draw  [loosely dash dot, orange]  (A) -- (E);
\draw  [loosely dotted] (A) -- (F);
\draw  [densely dotted, red]  (B) -- (C);
\draw  [blue, dashed]   (C) -- (D);
\draw  [green]   (D) -- (E);
\draw  [blue, dashed]  (E) -- (F);
\draw  [green]  (F) -- (B);

\end{tikzpicture}
\caption{The wheel graph on six vertices with the edges properly colored. There are five color classes: solid green, dashed blue, densely dotted red, loosely dotted black, and dotted dashed orange.}
\label{wheel3}
\end{figure}
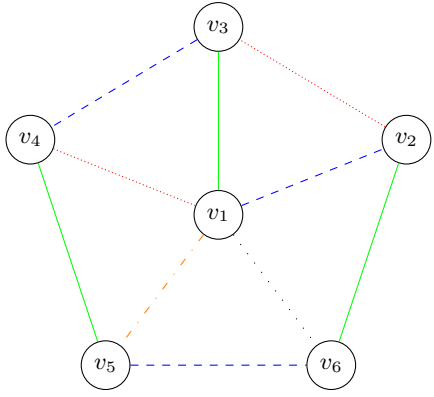

Let $G=(V,E)$ be a simple, undirected graph. In the maximum independent set problem, often denoted MaxIndSet, the goal is to find the largest set of independent vertices, or vertices that are not pairwise adjacent. This problem can be written as 

\begin{align}
    & \max\  \sum_{i \in V} x_i & \\
    & \mbox{s.t. } x_ix_j = 0 & \forall (i,j) \in E\\
    & x \in \{0,1\}^n &
\end{align}

\noindent $P \neq \{\emptyset\}$, as $|P| = |E|$, but for all $p_i$, $\underline{p_i} = 0$, so $\mathtt{feas}_i = 0$ for all $i$, and the new formulation is
$$\max_{x \in \{0,1\}^n} \sum_i x_i + \sum_{(i,j,)\in E} \lambda x_ix_j.$$

For the graph given in Fig. \ref{wheel3}, the resulting optimization function is:

\begin{align*}
    x_1 + x_2 + x_3 + x_4 + x_5 + x_6 + \\
    \lambda (x_1x_2 +  x_1x_3 +  x_1x_4 +  x_1x_5 +  x_1x_6 + \\
     x_2x_3 + x_3x_4 + x_4x_5 +  x_5x_6 +  x_2x_6)  
    \end{align*}


\textbf{Example: General optimization problem}\label{pubomgeneralsetup}

As a final example, consider 



\begin{align}\label{eq:reducedgeneralexample}
& \max \  \sum_{i\in [3]} x_i \\
& \mbox{s.t. }  x_1x_2+x_2x_3+2x_1x_3 \leq 3 \\
& x_i \in \{0,1\}
\end{align}

Now, $\underline{p_i} = 0$, so $\mathtt{feas} = 3$. Dualizing the first constraint, we get
\begin{align*}
& \max(x,\lambda) \sum_{i \in [3]} x_i + \lambda(x_1x_2+x_2x_3+2x_1x_3-\delta_{11}-2\delta_{12}-3)^2 \\
 & =\sum_{i \in [3]} x_i  + \lambda(-5x_1x_2+10x_1x_2x_3-5x_2x_3-8x_1x_3-2\delta_{11}x_1x_2- \\
 & 2\delta_{11}x_2x_3-4\delta_{11}x_1x_3-4\delta_{12}x_1x_2-4\delta_{12}x_2x_3- 7\delta_{11}+16\delta_{12}+ 4\delta_{11}\delta_{12}+9).
\end{align*}

\subsection{A QAOA Approach}\label{qaoa}


This section assumes we are optimizing a problem of the form $\min_{x \in \{0,1\}^n} \sum_{m_i \in M} m_i(x)$. The natural extension of QAOA on general PUBOs is to define the unitary operator $U(C,\gamma) = e^{-i \gamma C} = \prod_{m_i \in M} e^{-i\gamma m_i} = \prod_{m_i \in M} U(m_i, \gamma)$, while the mixing operator remains $U(B,\beta) = e^{-i \beta B}$ where $B = \Sigma_{v \in V(G)} B_v$, $B_v= \sigma_v^x$ for $v \in V(G)$ and $\sigma_v^x$ is the Pauli X operator acting on qubit $v$. $U(C,\gamma)$ can be compiled on a circuit by decomposing it into a sequence of gates performing all $U(m_i, \gamma)$ operators. The number of qubits each $U(m_i, \gamma)$ acts on is the number of variables in monomial $m_i$, which depends on the size of the support of $p_i$, denoted $supp(p_i)$. We seek to explore the relationship between the structure of the monomial and the minimum depth required for a quantum circuit to optimize such a function. 

First, we assume that all $m_i$ have been combined optimally to fit the hardware, meaning each polynomial has size at most the maximal gate size the hardware supports, and any monomials that can be combined and fit on one gate have been combined. Although current hardware currently supports gate width of two, we look at larger gate width for completeness. Operators $U(m_i, \gamma)$ and $U(m_j, \gamma)$ cannot be performed in parallel unless they act on disjoint sets of qubits. With that in mind, we construct a proper hypergraph edge coloring that minimizes the total depth of circuit, where edges of the same color represent sets of operators that can be performed in parallel. We let $H = (V_H,E_H)$ be such a hypergraph where $V_H = \{1,\ \ldots,\ n\}$ and $E_H$ consists of edges $e_i = supp(m_i)$ for all $m_i\in M$.

\begin{theorem}
Every proper edge coloring of $H$ corresponds to a valid circuit for $PUBO$, where the depth of the shallowest circuit is $\chi'(H)+1$.
\end{theorem}

\begin{proof}
Let $v_av_b...v_d$ be the support of monomial $m_i$ and $v_ev_f...v_h$ be the support of monomial $m_j$ such that $\{v_a, v_b, ... , v_d\} \cap \{v_e, v_f, ... , v_h\} = \{\emptyset\}$. Then, $U(C_{m_i}, \gamma)$ and $U(C_{m_j}, \gamma)$ can be implemented simultaneously in a circuit. Since the intersection is empty, the edges may receive the same color in a proper coloring, but do not necessarily, as there may be several proper colorings of one graph. Thus, a proper edge coloring gives a feasible implementation of a circuit that can be used to perform MaxCut. There exists a coloring of $H$ that uses exactly $\chi'(H)$ colors, and by definition, any coloring that uses fewer colors is not proper. If the coloring is not proper, two edges that share a vertex have the same color and their corresponding gates cannot be implemented simultaneously. Hence, the depth of the shallowest circuit is $\chi'(H)+1$, as one must be added to account for $U(B,\beta)$.
\end{proof}

Determining the chromatic index of hypergraphs in general is a difficult problem. In 1972, Erd\"{o}s, Faber, and Lov\'{a}sz conjectured that the chromatic index of any linear hypergraph on $n$ vertices is at most $n$ \cite{erdHos1975problems}. Since then, the conjecture has been proven if $H$ satisfies $\Delta(H) \leq \sqrt{n+\sqrt{n}+1}$ \cite{paul2012edge}. Additionally, Chang and Lawler showed that the chromatic index of a hypergraph $H$ on $n$ vertices is at most $\lceil 1.5n-2 \rceil$ with no restriction on the degrees of the vertices. In 1992, Kahn showed that $\chi'(H) \leq n + o(n)$ for linear $H$ \cite{kahn1992coloring}. Note that since any two edges in a linear hypergraph intersect in at most one vertex, that is equivalent to saying any two monomials in an optimization problem share at most one common variable. As there are bounds on the chromatic index of linear hypergraphs, in a general combinatorial optimization problem, one could attempt to relax the problem such that for any two monomials $a$ and $b$, $|supp(a) \cap supp(b)| \leq 1$ in order to have a rough bound on the depth of the circuit.

In addition to linear hypergraphs, there has been work on bounding the chromatic index of $k$-uniform hypergraphs. Pippenger and Spencer proved that if a $k$-uniform hypergraph has minimum degree asymptotic to the maximum degree and asymptotic codegree negligible compared to the maximum degree, then for some $\delta>0$, $\chi'(H) \leq (1+\delta)\Delta(H)$ \cite{pippenger1989asymptotic}. Later, Alon and Kim showed that if $H$ is $k$-uniform and if any two edges have at most $t$ vertices in common and maximum degree sufficiently large as a function of $k$, then $|E_h| \leq (t-1+\frac{1}{t})\Delta(H)$, which bounds the chromatic index of $H$ from above \cite{alon1997degree}. As each edge in a $k$-uniform hypergraph contains $k$ vertices, it is equivalent to the original combinatorial optimization problem containing monomials that consist of precisely $k$ variables. Thus, the circuit depth of problems that can be written such that each monomial has the same size support can be bounded.

We can potentially combine $U(C_a, \gamma)$ and $U(C_b, \gamma)$ into $U(C_{a,b},\gamma)$, which could reduce the number of colors needed for the corresponding graph. Doing so, however, requires solving a potentially difficult optimization problem. Consider the problem:

\begin{align}
\min & \sum_{c \in C} z_c \\
\mbox{s.t. } & \sum_{s \in S} x^c_s \leq |S| - 1 & \forall\  c \in C,\ S \ \mbox {s.t. } |\cup_{s \in S} s| > L\\
& x_e^c \leq z_c & \forall \ e \in E, \ c \in C\\
& x,\ z \in \{0,1\}
\end{align}
where $L$ is the number of qubits in the largest gate the hardware can perform, $c$ is a color in the collection of colors $C$, $s$ is an edge of $S \subset E(\mathcal{H})$, and $x_s^c$ indicates that $s$ receives color $c$ in a particular proper coloring.
One obvious example of how to combine gates is, if for monomials $a$ and $b$, $supp(a) \subset supp(b)$ then the size of the gate required for $U(C_{a,b},\gamma)$ will be identical to $U(C_{b},\gamma)$. Thus, $U(C_a, \gamma)$ and $U(C_{a,b}, \gamma)$ can be combined since $supp(a) \cup supp(b) = supp(b)$.

Throughout the rest of this paper, we define the \emph{derived graph} as the graph corresponding to a combinatorial optimization problem whose vertex set consists of the variables in the problem and whose edges connect vertices that are found in a common monomial. The \emph{derived hypergraph} is similarly defined. 




\subsubsection{Examples continued}

In this section, we analyze the structure of the resulting hypergraphs built from the examples in Section \ref{examplespartone} and discuss how this impacts the difficulty of performing each problem on NISQ devices.

\textbf{Example: MaxCut, continued} 

 The support of the cost function is six, but each gate acts on two qubits in the circuit since each monomial has at most two unique $x_i$ terms. We can define gates $U(C_{i,j}, \gamma)$ for monomials that have two variables, $x_i$ and $x_j$, and gates $U(C_k, \gamma)$ for monomials in one variable, $x_k$.

	\begin{figure}[h]
	\centering
	\includegraphics[scale=0.75]{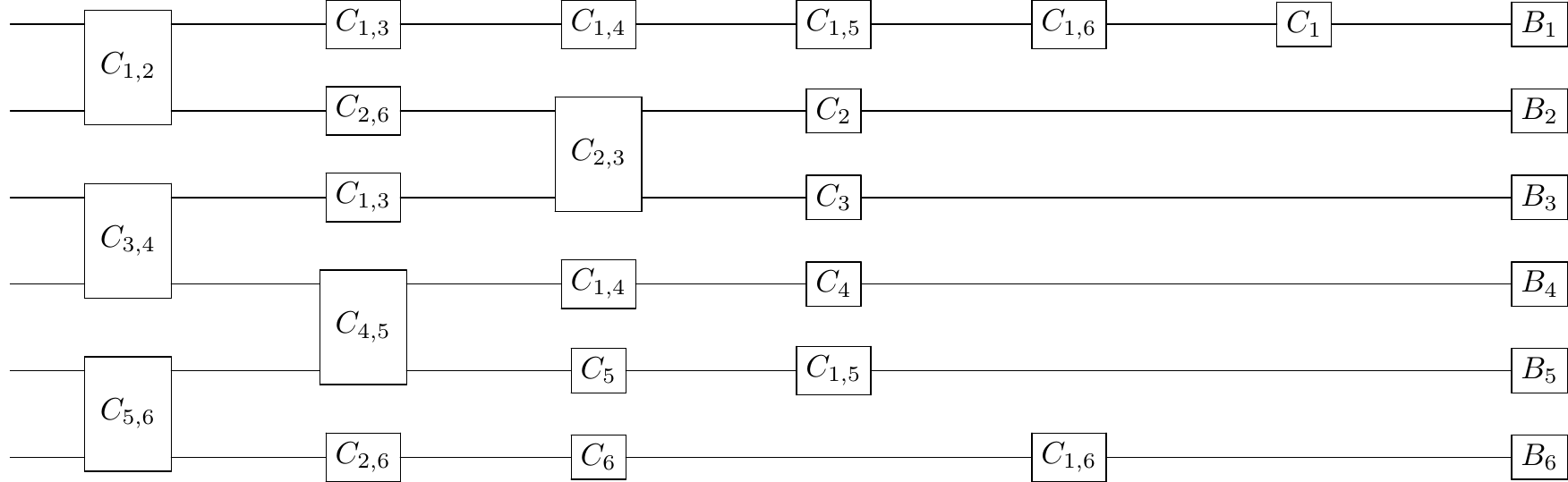}
	\caption{The circuit diagram for a $1$-level QAOA on the wheel graph on six vertices. We use the notation $C_{i,j}$ to represent $U(C_{i,j},\gamma)$ to make the image clearer to read. If two qubits are used in one gate but not next to each other in number order, the gate in the diagram appears split, but the pieces are in the same column.}
	\label{fig:pubocircuitdiagrammaxcut}
\end{figure}

 As $i, j, m,$ and $n$ must be unique in order to run $C_{i,j}$ and $C_{m,n}$ at the same time, we can color the edges of a graph $G$ and perform operators associated to the edges of the same color class at once. However, each gate $C_k$ can be run simultaneously, as each depends on precisely one qubit. Thus, the depth of the circuit for MaxCut is either $\chi '(G) + 1$ or $\chi '(G) + 2$, as one must be added to account for the $B$ gates, and the depth scales linearly with the number of iterations of the algorithm. Figure \ref{fig:pubocircuitdiagrammaxcut} is a circuit diagram for implementing MaxCut on $W_6$ using the PUBO mapping and QAOA approach, where the circuit has a depth of $\chi '(G) + 2$. 

\textbf{Example: Maximum Independent Set, continued}

The support of the optimization function has size six, and each monomial is comprised of at most two variables. The circuit diagram for this example is the same as in \textbf{Example: MaxCut, continued}, as it contains the same monomials, up to constants and signs.

\textbf{Example: General Optimization Problem, continued}

Since several monomials in the function to optimize are contained in the support of others, the gates needed in the QAOA circuit are those acting on $x_1x_2x_3, \ x_1x_2\delta_{11}, \ x_1x_2\delta_{12}, \ x_1x_3\delta_{11}, \ x_1x_3\delta_{12},$ $\ x_2x_3\delta_{11}, \ x_2x_3\delta_{12}$ and $\delta_{11}\delta_{12}$, and the associated hypergraph and coloring for it is Fig. \ref{fig:generalpuboexample}. The circuit diagram for this example is seen in Fig. \ref{fig:pubocircuitgeneral}.

	\begin{figure}
	\centering
	\begin{subfigure}{.5\textwidth}
  \centering
  \includegraphics[width=.7\linewidth]{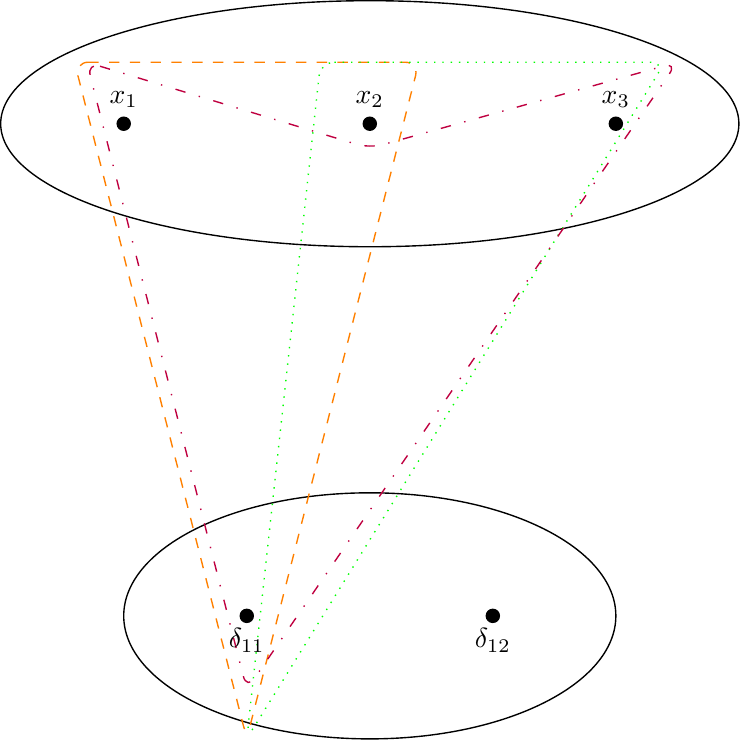}
  \label{fig:sub1}
\end{subfigure}%
\begin{subfigure}{.5\textwidth}
  \centering
  \includegraphics[width=.7\linewidth]{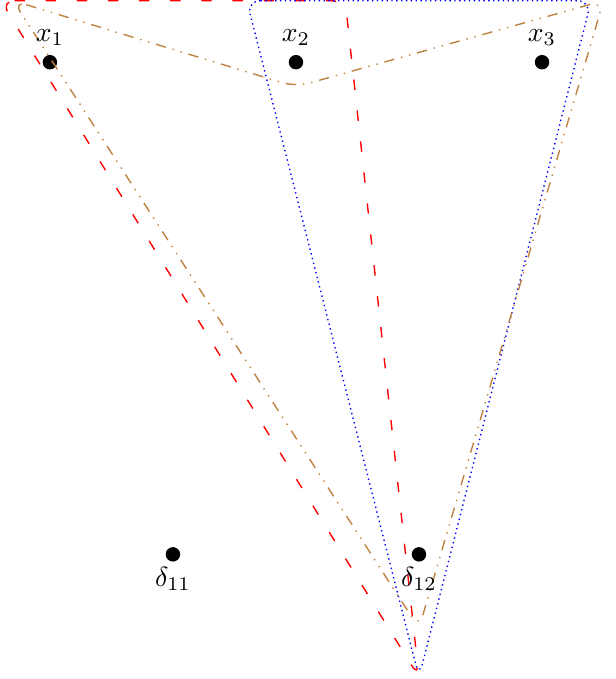}
  \label{fig:sub2}
\end{subfigure}
\caption{The coloring for the hypergraph in \textbf{Example: General Optimization Problem}. It has vertices $x_1$, $x_2$, $x_3$, $\delta_{11}$, and $\delta_{12}$. The edges have been placed into two separate images to show the coloring more clearly, though the entire hypergraph contains the edges found in both figures. No colors are repeated between the left and right sides, and the hypergraph requires seven colors.}
\label{fig:generalpuboexample}
\end{figure}

\begin{figure}[h]
	\centering
	\includegraphics[scale=0.65]{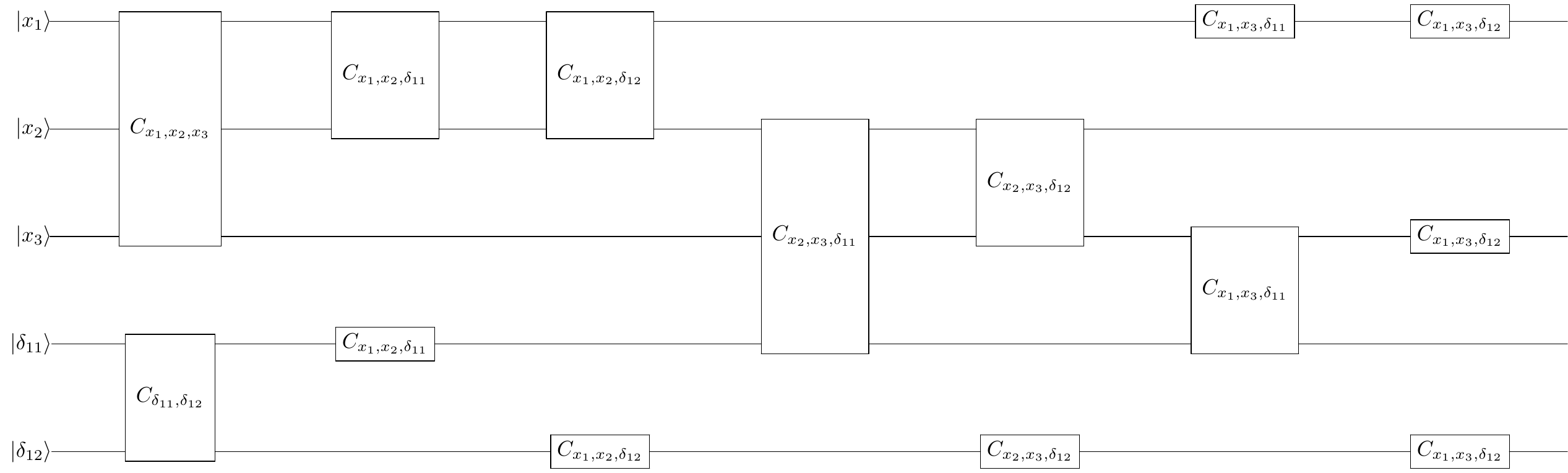}
	\caption{The circuit diagram for \textbf{Example: General Optimization Problem, continued}. We use  $C_{i,j,k}$ to represent $U(C_{i,j,k},\gamma)$ to make the image clearer to read. If multiple qubits are used in one gate but not next to each other in number order, the gate in the diagram appears split, but the pieces are in the same column.}
	\label{fig:pubocircuitgeneral}
\end{figure}

\section{QAOA circuit depth bounds for some combinatorial optimization problems}\label{problemanalysis}

In this section, we review some combinatorial optimization problems and discuss the depth of circuit required for one QAOA iteration of each problem instance. NP-complete problems can be reduced to other NP-complete problems, however the act of reducing may impact the depth of circuit, which may or may not be desirable depending on the hardware.

 \subsection{Vertex Covering}\label{vertexcoveringgeneralsetup}
 
 A vertex cover of a graph $G=(V,E)$ is a collection $S \subset V$ such that for all $xy \in E$, at least one of $x$ or $y$ is contained in $S$. Finding a minimum vertex covering is classically NP-complete \cite{karp1972reducibility} and is written as

\begin{align*}
    & \min\  \sum_{i \in V} x_i & \\
    & \mbox{s.t. } (1-x_i)+(1-x_j) \leq 1 & \forall (i,j) \in E\\
    & x_i \in \{0,1\} &
\end{align*}





As each constraint consists of the sum of two unique variables and a constant, dualizing them gives monomials consisting of at most two distinct variables, as each monomial corresponds to an edge in the graph. The derived graph, has vertex set $V= x_i \cup \delta_{ij}$ for $i \in [n]$ and $j \neq i$ with $j \in [n]$. Note that $\delta_{ij}$ is incident only to vertices $x_i$ and $x_j$, as it only occurs in constraints including those variables. Thus, the depth of circuit is  $2\chi(G)+1 $, so the difficulty of covering problems is directly related to the maximum degree of the problem. Graphs with low degree allow for a shallow circuit in one iteration of QAOA, so they should be suitable for NISQ devices. 


 \subsection{Knapsack}\label{knapsackgeneralsetup}
 
 In the knapsack problem, a collection of objects, $x_i$ for $i \in [n]$, are assigned a weight, $w_i$, and a value, $v_i$. The goal is to maximize the sum of the value of the objects while the sum of the weights of the objects is restricted to be less than some constant $W$. This problem is NP-complete classically, as well \cite{karp1972reducibility}. As an integer program, it is written

\begin{align*}
    & \max\  \sum_{i \in [n]} v_ix_i & \\
    & \mbox{s.t. } \sum_{i \in [n]} w_ix_i \leq W &\\
    & x_i \in \{0,1\} &
\end{align*}


 



Knapsack problems where $W$ is large pose problems for testing on quantum computers because the larger $W$ is, the more $\delta_{ij}$ variables are needed when dualizing the constraint.  The derived graph contains vertices for each $x$ variable and $\delta$ variable and is complete, meaning that its edge coloring has minimum coloring of at least $n + ln(W)$, however pre-processing can be used to reduce the depth of the embedding. Assuming the weights are  ordered such that $w_1 \leq w_2 \leq ... \leq w_n$, the total weight of an optimal solution must be at least $W - w_n$, since, if not, there is room in the knapsack for an additional item. With this in mind, the knapsack constraint can be written as $W - w_n  \leq  \sum_{i \in [n]} w_ix_i \leq W,$ which now requires $\ln(w_n)$ many additional variables and leads to a circuit depth of $n + ln(w_n)$. In order for there to be nontrivial instances of the knapsack problem with small $w_n$, there necessarily must be small $W$. Knapsack problems with small $W$ can be suitable for experimentation, however, classically, problems with bounded $W$ are polynomial, and can be easily solved by conventional computing by dynamic programming \cite{andonov2000unbounded, frieze1976shortest}. Thus, they may not be suitable for quantum computing.

 \subsection{Traveling Salesperson}\label{tspgeneralsetup}
 
 The traveling salesperson problem (TSP) can be viewed as a problem on a graph $G=(V,E)$ where each edge $e$ has an associated weight, $w_e$. The goal is to start in a vertex, say $v_1$, use edges to visit each vertex exactly once, and return to $v_1$, all while minimizing the sum of the weights of the edges used. This problem is classically NP-hard, but there exist some heuristics for the problem\cite{ouaarab2014discrete, masutti2009self}.
 
 Let $x_e$ represent if the salesperson travels along edge $e$.

 one formulation of the problem is  
 
 \begin{align*}
    & \min\  \sum_{e \in E} w_{e}x_{e} & \\
    & \mbox{s.t. } 0 \leq x_{e} \leq 1 \ \forall e \in E \\
    & \sum_{e \ni i } x_{e}=2 \ \forall i \in V\\
    & \sum_{e =(i,\ j),\ i \in Q,\ j \in Q} x_{e} \leq |Q|-1 & Q\subsetneq[n], |Q| \geq 2.
\end{align*}

The derived graph for TSP has vertex set $x_i \cup \delta_{ij}$, where there are two $\delta_{ij}$ variables per constraint. The edges form a complete graph on all $x_i$ vertices, and connect $\delta_{ij}$ to $\delta_{ik}$ for $j \neq k$. As there are two $\delta$ variables per constraint, they form a disjoint collection of edges. The rest of the edges connect every $x_i$ variable to every $\delta_{ij}$ variable. Thus, the maximum degree of the graph is $n-1$ plus twice the number of constraints. Denoting the number of constraints as $N_c$, the depth of circuit is $n-1+2N_c$, where $N_c$ can be large, depending on the problem instance. Similarly to knapsack problems, it can be difficult to implement TSP on NISQ devices because of the subtour constraint, $Q \subsetneq [n]$, and the fact that so many new variables are introduced in dualizing. 
%






  \subsection {SAT}\label{3satgeneralsetup}
  In Boolean satisfiability problems (SAT), there are a set of clauses, $C$, containing a set of literals, $N$. The goal is to determine if the values of TRUE or FALSE can be assigned to each literal in a clause such that it evaluates to TRUE. This problem, again, is classically NP-complete \cite{karp1972reducibility}, even when each clause contains only three literals. Let $\{z_c\}_{c \in C}$ be a collection of indicator variables for clauses in three variables, where $z_i=0$ if clause $i$ is satisfied and $1$ if not. Let $x_i$ be the indicator variable denoting if literal $i$ is satisfied. Let $TRUE_c$ ($FALSE_c$) be the set of literals that must be true (false) to satisfy clause $c$. Then, the problem can be written as
  
    \begin{align*}
    & \min\  \sum_{c \in C} z_c & \\
    & \mbox{s.t. } \sum_{x_{i}= TRUE_c } x_{i} +\sum_{x_{i}= FALSE_c} (1-x_{i}) \geq 1+z_c & \forall c \in C\\
    &x_i,\ z_c \in \{0,1\} 
\end{align*}

However, taking the contrapositive, we have $\sum_{x_{i} = TRUE} (1- x_{i}) + \sum_{x_{i} = FALSE} x_{i} \leq 2 +z_c$. The derived graph, is again, a graph consisting of all $x_i$ vertices and  two dummy variables, $\delta_{c}^1$ and $\delta_c^2$, per clause. If $x_i$ appears in the set of clauses $C_{x_i} \subset C$, the degree of the $x_i$, $d_{x_i}$, is $d_{x_i} =  |\cup_{c \in C_{x_i}}c|-1 + 2|C_{x_i}|$. SAT can be a good problem for NISQ devices if the set of literals is large while the number of literals in each clause and the number of clauses are relatively small, as this guarantees a literal cannot occur in many clauses and each literal does not appear in clauses with several others.




\section{Discussion}\label{conclusion}

We have shown how to map arbitrary combinatorial optimization problems to polynomial unconstrained binary optimization problems (PUBOs) by dualizing constraints, and applied the method to a few combinatorial optimization problems. Additionally, we discussed how use the PUBOs to derive a graph that represents problem instances and used this to show that the depth of the QAOA circuit needed to run the problem is $\chi'(G)+1$.  We then considered various combinatorial optimization problems and determined the depth of circuit required to use QAOA to solve them. In particular, since the Vertex Covering problem has a low depth of circuit, it appears to be suitable for NISQ devices, as do instances of SAT problems that have large sets of literals but few clauses and few literals in each clause. Due to the number of new variables that must be introduced to dualize Knapsack and TSP, they do not appear to be good problems to test on NISQ devices.

 Clearly, the maximum degree of a vertex affects the circuit depth in combinatorial optimization problems in which each monomial consists of at most two unique variables, such as MaxCut and MaxIndSet. Specifically, the depth of the QAOA circuit is $\chi'(G)+1$, where $\chi'(G)$ for graphs that are not hypergraphs is $\Delta$ or $\Delta+1$, by a classic result of Vizing \cite{vizing1964estimate}. In the case of monomials of at least three variables, a lower bound for the circuit depth is the number of colors needed in a proper edge coloring of the associated hypergraph, $H$, which is a difficult problem. A trivial lower bound on this number is the maximum degree of $H$ while a trivial upper bound is the number of edges in $H$.

The depth of circuit is hard to determine in part because when dualizing, squaring $p(x)$ can potentially yield monomials with larger support than any in $c$. Sparser constraints are preferable because the polynomials have smaller support, which decreases the size of each gate. However, sparser constraints does not imply a shallower circuit depth. For example. consider MaxCut on a star graph on $n$ vertices, that is a connected bipartite graph in which one part contains one vertex and the other contains $n-1$. The depth of circuit is $n$, as each edge must have a unique color.

As any combinatorial optimization problem can be mapped to a PUBO via dualizing constraints, we can examine the resulting QAOA circuit and bound the depth of it by the edge coloring of the hypergraph associated to the problem instance. Although current hardware is limited to two qubit gates, larger gates can be decomposed into two qubit gates. It would be interesting to see if there is a way to construct a graph associated to the decomposed gates and if its chromatic number, or some other property of the graph, determines the depth of circuit. 
\section*{Acknowledgements}

This work was supported by DARPA ONISQ program under award W911NF-20-2-0051. J. Ostrowski acknowledges the Air Force Office of Scientific Research award, AF-FA9550-19-1-0147.

This manuscript has been authored by UT-Battelle, LLC under Contract No. DE-AC05-00OR22725 with the U.S. Department of Energy. The United States Government retains and the publisher, by accepting the article for publication, acknowledges that the United States Government retains a non-exclusive, paid-up, irrevocable, world-wide license to publish or reproduce the published form of this manuscript, or allow others to do so, for United States Government purposes. The Department of Energy will provide public access to these results of federally sponsored research in accordance with the DOE Public Access Plan. (http://energy.gov/downloads/doe-public-access-plan).
\bibliographystyle{abbrv}
\bibliography{references}
\end{document}